\documentclass[letterpaper,fleqn,11pt]{article}
\usepackage[margin=1in]{geometry}
\usepackage{amsmath,amssymb,amsthm}
\usepackage{mathtools}
\usepackage{thm-restate}
\usepackage{array}
\usepackage{enumitem}
\usepackage{environ}
\usepackage{tikz}
\usetikzlibrary{decorations.pathreplacing,arrows.meta}
\usepackage{xcolor}
\usepackage{comment}
\usepackage[bookmarksnumbered=true,hypertexnames=false]{hyperref}
\usepackage{algorithm,algpseudocode}
\usepackage[capitalize,sort]{cleveref}

\newenvironment*{tightenum}{\enumerate[noitemsep]}{\endenumerate}
\newcommand*{\reason}[1]{\tag{#1}}
\newcommand*{\TCNLA}[1][]{}  %
\newcommand*{\ifNoAppendix}[1]{}
\newcommand*{\coecref}[1]{\cref{#1}}  %

\let\shortcite\cite

\def\colorschemesepia{sepia}
\def\colorschemedark{dark}
\def\colorschemelight{light}

\ifx\colorscheme\undefined
\let\colorscheme\colorschemelight
\fi

\ifx\colorscheme\colorschemelight
\colorlet{textColor}{black}
\colorlet{bgColor}{white}
\fi

\ifx\colorscheme\colorschemesepia
\definecolor{textColor}{HTML}{433423}
\definecolor{bgColor}{HTML}{fbf0da}
\fi

\ifx\colorscheme\colorschemedark
\definecolor{textColor}{HTML}{bdc1c6}
\definecolor{bgColor}{HTML}{202124}
\definecolor{textBlue}{HTML}{8ab4f8}
\definecolor{textRed}{HTML}{f9968b}
\definecolor{textGreen}{HTML}{81e681}
\definecolor{textPurple}{HTML}{c58af9}
\else
\colorlet{textBlue}{blue!50!black}
\colorlet{textRed}{red!50!black}
\colorlet{textGreen}{green!50!black}
\definecolor{textPurple}{HTML}{681da8}
\fi

\ifx\colorscheme\colorschemelight\else
\pagecolor{bgColor}
\color{textColor}
\fi

\allowdisplaybreaks

\let\eps\varepsilon
\newcommand*{\etal}{{\em et~al.}}
\newcommand*{\Th}{^{\textrm{th}}}
\newcommand*{\defeq}{:=}
\newcommand*{\wLoG}{without loss of generality}
\newcommand*{\WLoG}{Without loss of generality}

\DeclarePairedDelimiter\ceil{\lceil}{\rceil}
\DeclarePairedDelimiter\floor{\lfloor}{\rfloor}

\newcommand*{\Ahat}{\widehat{A}}
\newcommand*{\Ical}{\mathcal{I}}
\newcommand*{\Icalhat}{\widehat{\mathcal{I}}}
\newcommand*{\ellhat}{\widehat{\ell}}
\newcommand*{\vhat}{\widehat{v}}

\newcommand*{\Dtild}{\widetilde{D}}
\newcommand*{\Icaltild}{\widetilde{\mathcal{I}}}
\newcommand*{\Mtild}{\widetilde{M}}
\newcommand*{\Ntild}{\widetilde{N}}
\newcommand*{\vtild}{\widetilde{v}}
\newcommand*{\betatild}{\widetilde{\beta}}

\newcommand*{\mms}{\text{MMS}}
\newcommand*{\MMS}{\mathrm{MMS}}
\newcommand*{\mmsScore}{\mathrm{MMSscore}}
\newcommand*{\imp}{\min(\frac{1}{36}, \frac{3}{4(4n-1)})}
\newcommand*{\impText}{min(1/36, 3/(16n-4))}

\DeclareMathOperator{\approxMMS}{\mathtt{approxMMS}}
\DeclareMathOperator{\approxMMSHyp}{\hyperref[algo:gt]{\approxMMS}}

\DeclareMathOperator{\toOrdered}{\mathtt{toOrd}}
\DeclareMathOperator{\toOrderedHyp}{\hyperref[defn:toOrd]{\toOrdered}}
\DeclareMathOperator{\normalize}{\mathtt{normalize}}
\DeclareMathOperator{\normalizeHyp}{\hyperref[algo:normalize]{\normalize}}
\DeclareMathOperator{\reduce}{\mathtt{reduce}}
\DeclareMathOperator{\reduceHyp}{\hyperref[sec:valid-redn]{\reduce}}
\DeclareMathOperator{\bagFill}{\mathtt{bagFill}}
\DeclareMathOperator{\bagFillHyp}{\hyperref[algo:bagFill]{\bagFill}}

\newtheorem{lemma}{Lemma}
\newtheorem{theorem}{Theorem}

\newtheorem{definition}{Definition}

\newtheorem{example}{Example}

\crefname{claim}{Claim}{Claims}

\algnewcommand{\LineComment}[1]{\State \textcolor{gray}{\texttt{//} \textit{#1}}}

\NewEnviron{savedProof}[1]{%
\expandafter\expandafter\expandafter\gdef\expandafter\expandafter\csname sp:#1\endcsname\expandafter{%
\expandafter\proof\BODY\endproof}
\proof\BODY\endproof}
\newcommand*{\recallProof}[1]{\csname sp:#1\endcsname}

\newcommand*{\movedProofNote}{\begin{proof}Proof moved to \cref{sec:missing-proofs}.\end{proof}}
\newcommand*{\callMacro}[1]{\csname#1\endcsname}

\hypersetup{colorlinks,linkcolor=textRed,citecolor=textRed,urlcolor=textBlue}

\title{Simplification and Improvement of MMS Approximation%
\thanks{Jugal Garg and Eklavya Sharma were supported by NSF Grant CCF-1942321.}}

\date{\empty}

\newcommand*{\emailfont}[1]{\texttt{\small #1}}
\newcommand*{\affilUIUC}{University of Illinois at Urbana-Champaign, USA}
\newcommand*{\affilMPI}{Max Planck Institute for Informatics and
Graduiertenschule Informatik, Universit\"at des Saarlandes, Germany}
\newcommand*{\affilGrubhub}{Grubhub, USA}

\author{
Hannaneh Akrami\footnote{\affilMPI}\\\emailfont{hakrami@mpi-inf.mpg.de}
\and Jugal Garg\footnote{\affilUIUC}\\\emailfont{jugal@illinois.edu}
\and Eklavya Sharma\textsuperscript{\textdagger}\\\emailfont{eklavya2@illinois.edu}
\and Setareh Taki\footnote{\affilGrubhub}\\\emailfont{Staki@grubhub.com}
}

\begin{document}
\maketitle

\begin{abstract}
We consider the problem of fairly allocating a set of indivisible goods among $n$ agents with additive valuations, using the popular fairness notion of maximin share (MMS). Since MMS allocations do not always exist, a series of works provided existence and algorithms for approximate MMS allocations.
The Garg-Taki algorithm gives the current best approximation factor of $(\frac{3}{4} + \frac{1}{12n})$.
Most of these results are based on complicated analyses, especially those providing better than $2/3$ factor. Moreover, since no tight example is known of the Garg-Taki algorithm, it is unclear if this is the best factor of this approach. In this paper, we significantly simplify the analysis of this algorithm and also improve the existence guarantee to a factor of $(\frac{3}{4} + \min(\frac{1}{36}, \frac{3}{16n-4}))$. For small $n$, this provides a noticeable improvement. Furthermore, we present a tight example of this algorithm, showing that this may be the best factor one can hope for with the current techniques. 

\end{abstract}

\section{Introduction}
Fair division of a set of indivisible goods among $n$ agents with diverse preferences is a fundamental problem in many areas, including game theory, social choice theory, and multi-agent systems. 
We assume that agents have additive valuations. 
Maximin share (MMS) is one of the most popular fairness notion in this setting, introduced by Budish \shortcite{budish2011combinatorial}, which has attracted a lot of attention in recent years. It is preferred by participating agents over other notions, as shown in real-life experiments by~\cite{GatesGD20}. 
Every agent $i$ has an associated threshold, called her maximin share ($\MMS_i$), defined as the maximum value $i$ can get by partitioning the set of goods into $n$ bundles (one for each agent) and picking a lowest-value bundle.
An agent considers an allocation to be fair if she receives goods of total value at least her $\MMS$. 

A natural question is whether we can always find an allocation that gives each agent her MMS. Surprisingly, such an allocation need not always exist. 
Procaccia and Wang~\shortcite{procaccia2014fair} 
showed examples for any $n\ge 3$ in which MMS allocations do not exist. This motivated them to initiate the study of approximate MMS.
Agent $i$ considers an allocation to be $\alpha$-MMS fair to her for $\alpha\in(0,1)$ if she receives goods of total value at least $\alpha\cdot\MMS_i$.
They showed that a $2/3$-MMS allocation always exists. Ghodsi \etal\ \shortcite{ghodsi2018fair} improved this result by showing the existence of a $3/4$-MMS using a sophisticated algorithm with a very involved analysis.
More recently, Garg and Taki \shortcite{garg2021improved} improved this result to $(\frac{3}{4} + \frac{1}{12n})$-MMS
using a simple combinatorial algorithm, though their analysis remains quite involved. Furthermore, there is no tight example known for this algorithm, so it is unclear if this is the best factor of the approach. 

A complementary problem is to construct examples with the smallest upper bound on $\alpha$,
say $\alpha^*$, such that $\alpha$-MMS allocations do not always exist for $\alpha > \alpha^*$.
Feige, Sapir, and Tauber~\shortcite{feige2021tight} recently obtained the best-known $\alpha^* = 1-1/n^4$ for $n\ge 4$.
They also gave an improved value of $\alpha^* = 39/40$ for the special case of $n=3$ agents.
However, there is still a substantial gap between the lower and upper bounds. 

In this paper, we investigate the Garg-Taki algorithm and obtain the following results. 
\begin{itemize}
\item A significantly simple analysis of the algorithm. 
\item An improved bound of $(\frac{3}{4} + \imp)$-MMS by slightly modifying the algorithm. Since $\imp \ge \frac{1}{12n}$ for all $n \ge 3$, this provides noticeable improvement for small $n$. We note that $\frac{3}{4} + \frac{1}{12n}$ was the best-known bound for $n> 4$. 
\item A tight example of the Garg-Taki's and our algorithms, which shows the limits of this approach in obtaining a better bound of $\frac{3}{4}+O(1)$. Interestingly, this example only utilizes identical valuations, for which MMS allocations are known to exist.%
\footnote{We recently learned that, interestingly, this example was also obtained earlier
\cite{deuermeyer1982scheduling,babaioff2021fair} to show the tightness of
greedy algorithms for approximating MMS.}
\end{itemize}

Our simplified analysis not only helped us to improve the MMS bound but also, together with the tight example, shed more light on why and for which instances the algorithm cannot do better. We believe that these results would help reduce the gap further between the lower and upper bounds. 

\subsection{Related Work}
Computing the maximin share of any agent is NP-hard (even for 2 agents)%
\footnote{by a straightforward reduction from the partition problem.}, but 
a PTAS exists \cite{woeginger1997polynomial}.
Procaccia and Wang \shortcite{procaccia2014fair} showed the existence 
of a $2/3$-MMS allocation, which can also be computed in polynomial time for a constant $n$.
Later, the algorithm was modified \cite{amanatidis2017approximation,kurokawa2018fair}
to compute a $(2/3-\eps)$-MMS allocation in polynomial time (here $\eps > 0$ is a parameter of the algorithm, whose running time increases with $1/\eps$).
Barman and Krishnamurthy \shortcite{barman2020approximation} gave a simple greedy algorithm with an involved analysis to find a $\frac{2}{3}(1 + \frac{1}{3n-1})$-MMS allocation. %
Garg \etal\ \shortcite{garg2019approximating} gave a simple algorithm with a simple analysis
to output a $2/3$-MMS allocation.

Ghodsi \etal\ \shortcite{ghodsi2018fair} showed the existence of a $3/4$-MMS allocation 
using a complicated algorithm and analysis.
Garg and Taki \shortcite{garg2021improved} showed how to find
a $3/4$-MMS allocation in strongly polynomial time,
and showed that $(\frac{3}{4} + \frac{1}{12n})$-MMS allocations exist.
Their results use simple algorithms, but their analysis is still quite involved.
\smallskip

\noindent{\bf Special cases.}
Amanatidis \etal\ \shortcite{amanatidis2017approximation} showed that
when $m \le n+3$, an MMS allocation always exists.
Feige \etal\ \shortcite{feige2021tight} improved this to $m \le n+5$.
For $n=2$, MMS allocations always exist \cite{bouveret2016characterizing}.
For $n=3$, the MMS approximation was improved from $3/4$ \cite{procaccia2014fair}
to $7/8$ \cite{amanatidis2017approximation} to $8/9$ \cite{gourves2019maximin}, and then to $11/12$~\cite{feige2022improved}.
For $n=4$, Ghodsi \etal\ \shortcite{ghodsi2018fair} showed the existence of $4/5$-MMS. 
\smallskip

\noindent{\bf Experiments.} Bouveret and Lema{\^\i}tre \shortcite{bouveret2016characterizing} 
showed that MMS allocations usually exist (for data generated randomly using uniform or Gaussian valuations).
Amanatidis \etal\ \shortcite{amanatidis2017approximation} gave a simple and efficient algorithm and showed that when the valuation of each good is drawn independently and randomly from the uniform distribution on $[0, 1]$, the algorithm's output is an MMS allocation with high probability when the number of goods or agents is large. Kurokawa \etal\ \shortcite{kurokawa2016can} gave a similar result for arbitrary distributions of sufficiently large variance.
\smallskip

\noindent{\bf Chores.} MMS can be analogously defined for fair division of chores. MMS allocations do not always exist for chores \cite{aziz2017algorithms}, which motivated the study of approximate MMS~\cite{aziz2017algorithms,barman2020approximation,huang2021algorithmic}, with the current best approximation ratio being $11/9$. For 3 agents, $19/18$-MMS allocations exist~\cite{feige2022improved}. %
\smallskip

\noindent{\bf Other settings.} MMS has also been studied for non-additive valuations~\cite{barman2020approximation,ghodsi2018fair,li2021fair}. Generalizations have been studied where restrictions are imposed on the set of allowed allocations, like matroid constraints \cite{gourves2019maximin}, cardinality constraints \cite{biswas2018fair}, and graph connectivity constraints \cite{bei2022price,truszczynski2020maximin}. Stretegyproof versions of fair division have also been studied~\cite{barman2019fair,amanatidis2016truthful,amanatidis2017truthful,aziz2019strategyproof}. MMS has also inspired other notions of fairness, like weighted MMS \cite{farhadi2019fair}, AnyPrice Share (APS) \cite{babaioff2021fair}, Groupwise MMS \cite{barman2018groupwise,chaudhury2021little}, $1$-out-of-$d$ share \cite{hosseini2021guaranteeing}, and self-maximizing shares \cite{babaioff2022fair}.

\subsection{Outline of This Paper}

In \cref{sec:prelims}, we give formal definitions, notations, and preliminaries.
In \cref{sec:simple}, we give a very simple proof that (a minor modification of)
the Garg-Taki algorithm~\shortcite{garg2021improved} outputs a $3/4$-MMS allocation.
In \cref{sec:betterSummary}, we improve the analysis to show that the output
is a $(\frac{3}{4} + \imp)$-MMS allocation.
In \cref{sec:tight-example}, we give a tight example for our algorithm.

\section{Preliminaries}
\label{sec:prelims}

For any non-negative integer $n$, let $[n] \defeq \{1, 2, \ldots, n\}$.

A \emph{fair division instance} $\Ical$ is specified by a triple $(N, M, v)$,
where $N$ is the set of agents, $M$ is the set of goods,
and $v_{i,g}$ is the value of good $g \in M$ for agent $i \in N$.
For a set $S$ of goods, define $v_i(S) \defeq \sum_{g \in S} v_{i,g}$.
Then $v_i$ is called agent $i$'s \emph{valuation function}.
Intuitively, $v_i(S)$ is a measure of how valuable $S$ is to $i$.
For ease of notation, we write $v_i(g)$ instead of $v_i(\{g\})$.
We can assume \wLoG{} that $N = [n]$ and $M = [m]$, where $n = |N|$ and $m = |M|$
(though when dealing with multiple related fair division instances,
not making this assumption can sometimes simplify notation).

For a set $S$ of goods, let $\Pi_n(S)$ denote the set of partitions of $S$ into $n$ bundles.
For any valuation function $u$, define
\[ \MMS_u^n(S) \defeq \max_{X \in \Pi_n(S)} \min_{j=1}^n u(X_j). \]
When the fair division instance $(N, M, v)$ is clear from context,
we write $\MMS_i$ instead of $\MMS_{v_i}^{|N|}(M)$ for conciseness.

\subsection{Ordered Instance}
\label{sec:ordered}

\begin{definition}
A fair division instance $(N, M, v)$ is called \emph{ordered} if there is an ordering
$[g_1, g_2, \ldots, g_{|M|}]$ of goods $M$ such that
for each agent $i$, $v_{i,g_1} \ge v_{i,g_2} \ge \ldots \ge v_{i,g_{|M|}}$.
\end{definition}

We will now see how to reduce the problem of finding an $\alpha$-MMS allocation
to the special case of ordered instances.

\begin{definition}
\label{defn:toOrd}
For the fair division instance $\Ical \defeq (N, M, v)$,
$\toOrdered(\Ical)$ is defined as the instance $(N, [|M|], \vhat)$, where
for each $i \in N$ and $j \in [|M|]$, $\vhat_{i,j}$ is
the $j\Th$ largest number in the multiset $\{v_{i,g} \mid g \in M\}$.
\end{definition}

In Theorem 3.2 of \cite{barman2020approximation}, it was shown that
the transformation $\toOrdered$ is $\alpha$-\emph{MMS-preserving}, i.e.,
for a fair division instance $\Ical$, given an $\alpha$-MMS allocation of $\toOrdered(\Ical)$,
we can compute an $\alpha$-MMS allocation of $\Ical$ in polynomial time.
(The proof is based on ideas by Bouveret and Lema\^itre \shortcite{bouveret2016characterizing}).

\subsection{Valid Reductions}
\label{sec:valid-redn}

We use a technique called \emph{valid reduction},
that helps us reduce a fair division instance to a smaller instance.
This technique has been implicitly used in
\cite{bouveret2016characterizing,kurokawa2016can,kurokawa2018fair,amanatidis2017approximation,ghodsi2018fair,garg2019approximating}
and explicitly used in \cite{garg2021improved}.

\begin{definition}[Valid reduction]
In a fair division instance $(N, M, v)$, suppose we give the goods $S$ to agent $i$.
Then we are left with a new instance $(N \setminus \{i\}, M \setminus S, v)$.
Such a transformation is called a \emph{valid} $\alpha$-\emph{reduction} if
both of these conditions hold:
\begin{tightenum}
\item $v_i(S) \ge \alpha\MMS_{v_i}^{|N|}(M)$.
\item $\MMS_{v_j}^{|N|-1}(M \setminus S) \ge \MMS_{v_j}^{|N|}(M)$ for all $j \in N \setminus \{i\}$.
\end{tightenum}
\end{definition}

Note that valid reductions are $\alpha$-MMS-preserving, i.e.,
if $A$ is an $\alpha$-MMS allocation of an instance obtained by performing a valid reduction,
then we can get an $\alpha$-MMS allocation of the original instance by
giving goods $S$ to agent $i$ and allocating the remaining goods as per $A$.
A valid reduction, therefore, helps us reduce the problem of computing an $\alpha$-MMS
allocation to a smaller instance.

We now describe four standard transformations, called \emph{reduction rules},
and show that they are valid reductions.

\begin{definition}[Reduction rules]
\label{defn:redn-rules}
Consider an ordered fair division instance $(N, M, v)$, where $M \defeq \{g_1, \ldots, g_{|M|}\}$
and $v_{i,g_1} \ge \ldots \ge v_{i,g_{|M|}}$ for every agent $i$. Define
\begin{tightenum}
\item $S_1 \defeq \{g_1\}$.
\item $S_2 \defeq \{g_{|N|}, g_{|N|+1}\}$ if $|M| \ge |N|+1$, else $S_2 \defeq \emptyset$.
\item $S_3 \defeq \{g_{2|N|-1}, g_{2|N|}, g_{2|N|+1}\}$ if $|M| \ge 2|N|+1$, else $S_3 \defeq \emptyset$.
\item $S_4 \defeq \{g_1, g_{2|N|+1}\}$ if $|M| \ge 2|N|+1$, else $S_4 \defeq \emptyset$.
\end{tightenum}

Reduction rule $R_k(\alpha)$: If $v_i(S_k) \ge \alpha\MMS_i$ for some agent $i$,
then give $S_k$ to $i$.

A fair division instance is called $R_k(\alpha)$-\emph{irreducible} if $R_k(\alpha)$
cannot be applied, i.e., $v_i(S_k) < \alpha\MMS_i$ for every agent $i$
(otherwise it is called $R_k(\alpha)$-\emph{reducible}).
An instance is called \emph{totally}-$\alpha$-\emph{irreducible} if it is
$R_k(\alpha)$-irreducible for all $k \in [4]$.
\end{definition}

\begin{lemma}[Lemma~3.1 in \protect\cite{garg2021improved}]
For an ordered instance and for $\alpha \le 1$,
$R_1(\alpha)$, $R_2(\alpha)$, and $R_3(\alpha)$ are valid $\alpha$-reductions.
For an ordered instance and for $\alpha \le 3/4$,
if the instance is $R_1(\alpha)$-irreducible and $R_3(\alpha)$-irreducible,
then $R_4(\alpha)$ is a valid $\alpha$-reduction.
\end{lemma}

\begin{lemma}
\label{thm:vr-upper-bounds}
Let $\Ical \defeq ([n], [m], v)$ be an ordered instance where
$v_{i,1} \ge \ldots \ge v_{i,m}$ for each agent $i$.
For any $k \in [3]$, if $\Ical$ is $R_k(\alpha)$-irreducible, then
for each agent $i$ and every good $j > (k-1)n$, we have $v_{i,j} < \alpha\MMS_i/k$.
\end{lemma}
\begin{proof}
Since $\Ical$ is $R_k(\alpha)$-irreducible, we get $v_i(S_k) < \alpha\MMS_i$ for each agent $i$.
Let $t \defeq (k-1)n + 1$. Then
\[ \alpha\MMS_i > v_i(S_k) = \sum_{g \in S_k} v_{i,g}
\ge |S_k|\min_{g \in S_k} v_{i,g} = k v_{i,t}. \]
Hence, $\forall j \ge t$, we have $v_{i,j} \le v_{i,t} < \alpha\MMS_i/k$.
\end{proof}

\begin{lemma}
\label{thm:r1-irr-2n}
If an ordered instance $(N, M, v)$ is $R_1(\alpha)$-irreducible
for any $\alpha \le 1$, then $|M| \ge 2|N|$.
\end{lemma}
\begin{proof}
Assume $|M| < 2|N|$. Pick any agent $i \in N$. Let $P$ be an MMS partition of agent $i$.
Then some bundle $P_j$ contains a single good $\{g\}$.
Then $v_{i,g} = v_i(P_j) \ge \MMS_i$.
Hence, the instance is not $R_1(\alpha)$-irreducible for any $\alpha \le 1$.
This is a contradiction. Hence, $|M| \ge 2|N|$.
\end{proof}

We would like to convert fair division instances into totally-$\alpha$-irreducible instances.
This can be done using a very simple algorithm, which we call $\reduce_{\alpha}$.
This algorithm works for $\alpha \le 3/4$.
It takes an ordered fair division instance as input and repeatedly applies
the reduction rules $R_1(\alpha)$, $R_2(\alpha)$, $R_3(\alpha)$, and $R_4(\alpha)$
until the instance becomes totally-$\alpha$-irreducible.
The reduction rules can be applied in arbitrary order, except that
$R_4(\alpha)$ is only applied when $R_1(\alpha)$ and $R_3(\alpha)$ are inapplicable.

Note that the application of reduction rules changes the number of agents and goods,
which affects subsequent reduction rules.
More precisely, the sets $S_1$, $S_2$, $S_3$, $S_4$ (as defined in \cref{defn:redn-rules})
can change after applying a reduction rule.
So, for example, it is possible that an instance is $R_2(\alpha)$-irreducible,
but after applying $R_3(\alpha)$, the resulting instance is $R_2(\alpha)$-reducible.

\subsection{Normalized Instance}
\label{sec:normalized}

\begin{definition}[Normalized instance]
A fair division instance $(N, M, v)$ is called \emph{normalized} if for every agent $i$,
there is a partition $P^{(i)} \defeq (P^{(i)}_1, \ldots, P^{(i)}_{|N|})$ of $M$
such that $v_i(P^{(i)}_j) = 1$ $\forall j \in N$.
\end{definition}

Note that for a normalized instance, every agent's MMS value is 1.
Furthermore, for each agent $i$ and for every MMS partition $Q$ of agent $i$,
we have $v_i(Q_j) = 1$ $\forall j \in N$,
since each partition has total value at least 1
and $\sum_{j \in N} v_i(Q_j) = v_i(M) = \sum_{j \in N} v_i(P^{(i)}_j) = |N|$.

The algorithm $\normalizeHyp$ (c.f.~\cref{algo:normalize})
converts a fair division instance to a normalized instance.

\begin{algorithm}[tb]
\caption{$\normalize((N, M, v))$}
\label{algo:normalize}
\begin{algorithmic}[1]
\For{$i \in N$}
    \State Compute agent $i$'s MMS partition $P^{(i)}$.
    \State $\forall j \in N$, $\forall g \in P^{(i)}_j$,
        let $\vhat_{i,g} \defeq v_{i,g} / v_i(P^{(i)}_j)$.
\EndFor
\State \Return $(N, M, \vhat)$.
\end{algorithmic}
\end{algorithm}

\begin{lemma}
\label{thm:normalize}
Let $(N, M, \vhat) = \normalize((N, M, v))$. Then for any allocation $A$,
$v_i(A_i) \ge \vhat_i(A_i)\MMS_{v_i}^{|N|}(M)$ for all $i \in N$.
\end{lemma}
\begin{proof}
Let $\beta_i \defeq \MMS_{v_i}^n(M)$.
For any good $g \in P^{(i)}_j$,
$\vhat_{i,g} = v_{i,g} / v_i(P^{(i)}_j) \le v_{i,g} / \beta_i$.
Hence, $v_{i,g} \ge \vhat_{i,g}\beta_i$.
Hence, $v_i(A_i) \ge \vhat_i(A_i)\beta_i$.
\end{proof}

\Cref{thm:normalize} implies that $\normalize$ is $\alpha$-MMS-preserving,
since if $A$ is an $\alpha$-MMS allocation for the normalized instance $(N, M, \vhat)$,
then $A$ is also an $\alpha$-MMS allocation for the original instance $(N, M, v)$.

\section{\boldmath Simple Proof for Existence of \texorpdfstring{$3/4$-\mms}{3/4-MMS}~Allocations}
\label{sec:simple}

\newcommand{\ifBetter}[2]{#2}
\newcommand{\addIfBetterP}[2]{\ifBetter{\left(#1#2\right)}{#1}}
\newcommand{\BFn}{\ifBetter{n'}{n}}  %
\newcommand{\secP}{\ifBetter{B}{S}}  %

We give an algorithm, called $\approxMMSHyp$ (c.f.~\cref{algo:gt}),
that takes as inputs a fair division instance and an approximation factor $\alpha$,
and outputs an $\alpha$-MMS allocation. It works in three major steps:

\begin{algorithm}[tb]
\caption{$\approxMMS(\Ical, \alpha)$
\\ \textbf{Input:} Fair division instance $\Ical = (N, M, v)$ and approximation factor $\alpha$.
\\ \textbf{Output:} Allocation $A = (A_1, \ldots, A_n)$.}
\label{algo:gt}
\begin{algorithmic}[1]
\State $\Icalhat = \toOrderedHyp(\normalizeHyp(\reduceHyp_{\alpha}(\toOrderedHyp(\Ical))))$
\State $\Ahat = \bagFillHyp(\Icalhat, \alpha)$.
\State Use $\Ahat$ to compute an allocation $A$ for $\Ical$ with the same MMS approximation as $\Ahat$.
(This can be done since \cref{sec:ordered,sec:valid-redn,sec:normalized} show that
$\toOrdered$, $\reduce_{\alpha}$, and $\normalize$ are $\alpha$-MMS-preserving.)
\State \Return $A$
\end{algorithmic}
\end{algorithm}

\begin{enumerate}
\item Reduce the problem of finding an $\alpha$-MMS allocation to the special case
    where the instance is Ordered, Normalized, and totally-$\alpha$-Irreducible (ONI).
\item Compute an $\alpha$-MMS allocation for this special case
    using the $\bagFillHyp$ algorithm (c.f.~\cref{algo:bagFill}).
\item Convert this allocation for the special case to an
    allocation for the original fair division instance.
\end{enumerate}

We describe steps 1 and 3 in \cref{sec:simple:oni-reduce} and step 2 in \cref{sec:simple:bagFill}.
In this section, we only consider the case where $\alpha = 3/4$.
In \cref{sec:betterSummary}, we slightly modify $\approxMMS$ so that
it works for $\alpha = \frac{3}{4} + \imp$.
Our algorithm $\approxMMS$ is almost the same as the algorithm of Garg and Taki \shortcite{garg2021improved}.
The only difference is that, unlike them, we ensure that the output of step 1 is normalized.

\subsection{Obtaining an Ordered Normalized Irreducible (ONI) Instance}
\label{sec:simple:oni-reduce}

\begin{lemma}
\label{thm:oni-reduce}
Let $\Ical$ be a fair division instance.
Let $\Icalhat \defeq \toOrderedHyp(\normalizeHyp(\reduceHyp_{3/4}(\toOrderedHyp(\Ical))))$.
Then $\Icalhat$ is ordered, normalized, and totally-$3/4$-irreducible.
Furthermore, the transformation of $\Ical$ to $\Icalhat$ is $3/4$-MMS-preserving,
i.e., a $3/4$-MMS allocation of $\Icalhat$ can be used to obtain
a $3/4$-MMS allocation of $\Ical$.
\end{lemma}
\begin{proof}
Let $\Ical^{(1)} \defeq \toOrdered(\Ical)$.
Then $\Ical^{(2)} \defeq \reduce_{3/4}(\Ical^{(1)})$ is totally-$3/4$-irreducible and ordered,
since the application of reduction rules preserves orderedness.

Let $\Ical^{(3)} \defeq \normalize(\Ical^{(2)})$.
By \cref{thm:normalize}, $\normalize$ does not increase the ratio of a good's value to the MMS value.
Hence, $\Icalhat$ is totally-$3/4$-irreducible.
$\Icalhat$ is also normalized, since for each agent, $\toOrdered$ only changes
the identities of the goods, but the (multi-)set of values of the goods remains the same.
Hence, $\Icalhat$ is ordered, normalized, and totally-$3/4$-irreducible.

Since $\toOrdered$, $\reduce_{3/4}$, and $\normalize$ are $3/4$-MMS-preserving operations,
their composition is also $3/4$-MMS-preserving.
\end{proof}

The order of operations is important here,
as well as the need to call $\toOrdered$ twice,
since $\reduce$ requires the input to be ordered,
$\reduce$ may not preserve normalizedness,
and $\normalize$ may not preserve orderedness.

Garg and Taki \shortcite{garg2021improved} transform the instance as
$\reduce_{3/4}(\toOrdered(\Ical))$, since they do not need the input to be normalized.

\subsection{\boldmath \texorpdfstring{$3/4$}{3/4}-MMS Allocation of ONI Instance}
\label{sec:simple:bagFill}

Let $([n], [m], v)$ be a fair division instance that is
ordered, normalized, and totally-$3/4$-irreducible (ONI).
\WLoG{}, assume that $v_{i,1} \ge v_{i,2} \ge \ldots \ge v_{i,m}$ for each agent $i$.

Our algorithm, called $\bagFillHyp(\Ical, \alpha)$, creates $n$ bags,
where the $j\Th$ bag contains goods $\{j, 2n + 1 - j\}$.
(To create bags in this way, there must be at least $2n$ goods.
This is ensured by \cref{thm:r1-irr-2n}.)
It then repeatedly adds a good to an arbitrary bag,
and as soon as some agent $i$ values a bag more than $\alpha$,
that bag is allocated to $i$.
The algorithm terminates when all agents have been allocated a bag.
See \cref{algo:bagFill} for a more precise description.
(In this section, we set $\alpha = 3/4$.
In \cref{sec:betterSummary}, we set $\alpha = \frac{3}{4} + \imp$.)
$\bagFill$ computes a partial allocation, i.e.,
some goods may remain unallocated. But that can be easily fixed
by arbitrarily allocating those goods among the agents.

\begin{algorithm}[tb]
\caption{$\bagFill(\Ical, \alpha)$
\\ \textbf{Input:} Ordered instance $\Ical = ([n], [m], v)$ with $m \ge 2n$ and approximation factor $\alpha$.
\\ \textbf{Output:} (Partial) allocation $A = (A_1, \ldots, A_n)$.
}
\label{algo:bagFill}
\begin{algorithmic}[1]
\For{$k \in [n]$}
    \State $B_k = \{k, 2n+1-k\}$.
\EndFor
\State $U_G = [m] \setminus [2n]$  \Comment{unassigned goods}
\State $U_A = [n]$  \Comment{unsatisfied agents}
\State $U_B = [n]$  \Comment{unassigned bags}
\While{$U_A \neq \emptyset$}
    \Comment{loop invariant: $|U_A| = |U_B|$}
    \If{$\exists i \in U_A$, $\exists k \in U_B$, such that $v_i(B_k) \ge \alpha$}
        \LineComment{assign the $k\Th$ bag to agent $i$:}
        \State $A_i = B_k$
        \State $U_A = U_A \setminus \{i\}$
        \State $U_B = U_B \setminus \{k\}$
    \ElsIf{$U_G \neq \emptyset$}
        \State $g$ = arbitrary good in $U_G$
        \State $k$ = arbitrary bag in $U_B$
        \LineComment{assign $g$ to the $k\Th$ bag:}
        \State $B_k = B_k \cup \{g\}$.
        \State $U_G = U_G \setminus \{g\}$
    \Else
        \State \label{alg-line:bagFill:error}\textbf{error}: we ran out of goods.
            \Return \texttt{null}.
    \EndIf
\EndWhile
\State \Return $(A_1, \ldots, A_n)$
\end{algorithmic}
\end{algorithm}

$\bagFill(\Ical, \alpha)$ allocates a bag $B_k$ to agent $i$ only if $v_i(B_k) \ge \alpha$.
Hence, to prove that $\bagFill(\Ical, 3/4)$ returns a $3/4$-MMS allocation,
it suffices to show that $\bagFill$ terminates successfully,
i.e., line \ref{alg-line:bagFill:error} is never executed.

For $k \in [n]$, let $B_k \defeq \{k, 2n+1-k\}$ be the initial contents of the $k\Th$ bag
and $B'_k$ be the $k\Th$ bag's contents after $\bagFill$ terminates.
We consider two groups of agents.
Let $N^1$ be the set of agents who value all the initial bags at most $1$.
Formally, $N^1 \defeq \{i \in [n] \mid \forall k \in [n], v_i(B_k) \le 1\}$.
Let $N^2 \defeq [n] \setminus N^1 = \{i \in [n] \mid \exists k \in [n]: v_i(B_k) > 1\}$
be the rest of the agents.

Let $U_A$ be the set of agents that did not receive a bag when $\bagFill$ terminated.
Note that $U_A$ is non-empty iff we execute line \ref{alg-line:bagFill:error}.
We first show that all agents in $N^1$ receive a bag, i.e., $U_A \cap N^1 = \emptyset$.
Then we show that $U_A \cap N^2 = \emptyset$.
Together, these facts establish that $\bagFill$ terminates successfully,
and hence its output is $3/4$-MMS.

\begin{lemma}
\label{thm:pair-bound}
Let $([n], [m], v)$ be an ordered and normalized fair division instance.
For all $k \in [n]$ and agent $i \in [n]$, if $v_{i,k} + v_{i,2n-k+1} > 1$,
then $v_{i,2n-k+1} \le 1/3$ and $v_{i, k} > 2/3$.
\end{lemma}
\begin{proof}
It suffices to prove $v_{i,2n-k+1} \le 1/3$ and then $v_{i, k} > 2/3$ follows.
Let $P = (P_1, \ldots, P_n)$ be an MMS partition of agent $i$.
For $j \in [k]$ and $j' \in [2n+1-k]$,
$v_{i,j} + v_{i,j'} \ge v_{i,k} + v_{i,2n+1-k} > 1$, since the instance is ordered.
Furthermore, $j$ and $j'$ cannot be in the same bundle in $P$, since the instance is normalized.
In particular, no two goods from $[k]$ are in the same bundle in $P$.
Hence, assume \wLoG{} that $j \in P_j$ for all $j \in [k]$.

For all $j \in [k]$ and $j' \in [2n-k+1]$, $j' \not\in P_j$.
Thus, $\{k+1, \ldots, 2n-k+1\} \subseteq P_{k+1} \cup \ldots \cup P_n$.
By pigeonhole principle, there exists a bundle $B \in \{P_{k+1}, \ldots, P_{n}\}$
that contains at least $3$ goods $g_1, g_2, g_3$ in $\{k+1, \ldots, 2n-k+1\}$. Hence,
\begin{align*}
v_{i, 2n-k+1} &\le \min_{g \in \{g_1, g_2, g_3\}} v_{i,g}
\le \frac{1}{3}\sum_{g \in \{g_1, g_2, g_3\}} v_{i,g}
\TCNLA\le \frac{v_i(B)}{3} = \frac{1}{3}.
\qedhere \end{align*}
\end{proof}

\begin{lemma}
\label{thm:S:upper-1}
Let $i$ be any agent.
For all $k \in [n]$, if $v_i(B_k) \le 1$, then $v_i(B'_k) \le 1$.
\end{lemma}
\begin{savedProof}{upper-1}
If $B'_k = B_k$, then the claim obviously holds. Now assume $B_k \subsetneq B'_k$.
Let $g$ be the last good that was added to $B'_k$.
We have $v_i(B'_k \setminus g) < 3/4\ifBetter{ + \delta}{}$, otherwise $g$ would not be added to $B'_k$.
Also note that $g > 2\BFn$ and hence $v_{i, g} < 1/4\ifBetter{ + \delta/3}{}$
by \cref{thm:vr-upper-bounds}. Thus, we have
\begin{align*}
v_i(B'_k) &= v_i(B'_k \setminus g) + v_{i, g}
\ifBetter{\TCNLA}{}< \addIfBetterP{\frac{3}{4}}{ + \delta} + \addIfBetterP{\frac{1}{4}}{ + \frac{\delta}{3}}
= 1\ifBetter{ + \frac{4\delta}{3}}{}.
\qedhere
\end{align*}
\end{savedProof}

\begin{lemma}
\label{thm:S:n1-agents}
$U_A \cap N^1 = \emptyset$, i.e., every agent in $N^1$ gets a bag.
\end{lemma}
\begin{proof}
For the sake of contradiction, assume $U_A \cap N^1 \neq \emptyset$.
Hence, $\exists i \in U_A \cap N^1$. Also, for some $j \in [n]$, the $j\Th$ bag is unallocated.
Hence, $v_i(B_j') < 3/4$ and
\begin{align*}
n &= v_i(M) = v_i(B_j') + \sum_{k \in [n]\setminus \{j\}} v_i(B'_k)
    \reason{since $M = \bigcup_{k \in [n]} B'_k$}
\\ &< (n-1) + \frac{3}{4} = n - \frac{1}{4},  \reason{by \cref{thm:S:upper-1}}
\end{align*}
which is a contradiction. Hence, $U_A \cap N^1 = \emptyset$.
\end{proof}

Now we prove that $\bagFill$ allocates a bag to all agents in $N^2$,
i.e., $U_A \cap N^2 = \emptyset$.

\begin{lemma}
\label{thm:S:upper-1-12}
$i \in N^2 \implies v_{i, 2n+1} < 1/12$.
\end{lemma}
\begin{savedProof}{upper-1-12}
Since $i \in N^2$, there exists a bag $B_k$ such that $v_i(B_k) > 1$.
By \cref{thm:pair-bound}, $v_{i,k} > 2/3$. Thus, $v_{i,1} > 2/3$.
Moreover,
\begin{align*}
v_{i, 2\BFn+1} &< \frac{3}{4}\ifBetter{ + \delta}{} - v_{i,1}
    \reason{since $R_4(3/4\ifBetter{+\delta}{})$ is not applicable}
\\ &< \frac{3}{4}\ifBetter{ + \delta}{} - \frac{2}{3}
= \frac{1}{12}\ifBetter{ + \delta}{}.  \reason{since $v_{i,1} > 2/3$}
\end{align*}
\end{savedProof}

From now on assume for the sake of contradiction that $U_A \neq \emptyset$.
Let $a$ be a fixed agent in $U_A$. By \cref{thm:S:n1-agents}, $a \in N^2$.
Let $A^+ \defeq \{k \in [n] \mid v_a(B_k) > 1\}$,
$A^- \defeq \{k \in [n] \mid v_a(B_k) < 3/4\}$,
and $A^0 \defeq \{k \in [n] \mid 3/4 \le v_a(B_k) \le 1\}$.
We will try to get upper bounds on $v_a(B_k')$ for each of the cases
$k \in A^+$, $k \in A^-$, and $k \in A^0$.

Note that $n = |A^+| + |A^-| + |A^0|$.
Also, $n \in A^-$ since the instance is $R_2(3/4)$-irreducible,
and $|A^+| \ge 1$ since $a \in N^2$.

\begin{lemma}
\label{thm:S:upper-5-6}
$\forall k \in A^-$, $v_a(B'_k) < 5/6$.
\end{lemma}
\begin{savedProof}{upper-5-6}
If $B'_k = B_k$, then $v_a(B'_k) < 3/4\ifBetter{+\delta}{} < 5/6\ifBetter{+2\delta}{}$.
Otherwise, let $g$ be the last good that was added to $B'_k$.
Then $v_a(B'_k \setminus \{g\}) < 3/4\ifBetter{+\delta}{}$, otherwise $\bagFill$ would
assign $B'_k \setminus \{g\}$ to agent $i$ instead of adding $g$ to it. Hence,
\begin{align*}
v_a(B'_k) &= v_a(B'_k \setminus \{g\}) + v_{a,g}
\\ &< \addIfBetterP{\frac{3}{4}}{ + \delta} + v_{a,2\BFn+1}
    \reason{since $v_a(B'_k \setminus g) < 3/4\ifBetter{+\delta}{}$
        and $v_{a,g} \le v_{a,2\BFn+1}$}
\\ &< \addIfBetterP{\frac{3}{4}}{ + \delta} + \addIfBetterP{\frac{1}{12}}{ + \delta}
    = \frac{5}{6}\ifBetter{ + 2\delta}{}.
    \reason{$v_{a, 2\BFn+1} < 1/12\ifBetter{+\delta}{}$ by \cref{thm:\secP:upper-1-12}}
\end{align*}
\end{savedProof}

Let $\ell$ be the smallest such that for all $k \in [\ell+1, n]$,
$v_{a,k} + v_{a, 2n-k+1 + \ell} \le 1$.
See \cref{fig:ell} for a better understanding of $\ell$.
Note that $\ell \ge 1$, since $a \in N^2$.

\begin{figure*}[tb]
\centering
\newlength{\cellW}\newlength{\cellH}
\setlength{\cellW}{3.45em}
\setlength{\cellH}{1.8em}
\begin{tikzpicture}[
outerBox/.style = {semithick},
innerBord/.style = {},
highlight/.style = {very thick,fill={textColor!5!bgColor}},
myArrow/.style={->,>={Stealth},thick},
]
\draw[outerBox] (0, 0) rectangle +(12\cellW, 2\cellH);
\draw[innerBord] (0, 1\cellH) -- +(12\cellW, 0);
\foreach \x in {1,2,3.5,5.5,7,8.5,10,11}
    \draw[innerBord] (\x\cellW, 0) -- +(0, 2\cellH);
\draw[highlight] (3.5\cellW, 1\cellH) rectangle +(2\cellW, 1\cellH);
\draw[highlight] (7\cellW, 0) rectangle +(1.5\cellW, 1\cellH);
\foreach \x/\w/\downText/\upText in {
        0/1/1/2n,
        1/1/2/2n-1,
        2/1.5/\cdots/\cdots,
        3.5/2/k-\ell/2n+1-k+\ell,
        5.5/1.5/\cdots/\cdots,
        7/1.5/k/2n+1-k,
        8.5/1.5/\cdots/\cdots,
        10/1/n-1/n+2,
        11/1/n/n+1
        } {
    \path (\x\cellW, 0) rectangle +(\w\cellW, 1\cellH) node[pos=0.5] {$\downText$};
    \path (\x\cellW, 1\cellH) rectangle +(\w\cellW, 1\cellH) node[pos=0.5] {$\upText$};
}
\node[circle,draw,inner sep=0] (plus) at (6.5\cellW, -0.6\cellH) {$+$};
\draw[thick] (5\cellW, 1\cellH) -- (plus);
\draw[thick] (7.75\cellW, 0) -- (plus);
\draw[myArrow] (plus) -- +(0, -0.8\cellH);
\node (le1) at (6.5\cellW, -1.7\cellH) {$\le 1$};
\end{tikzpicture}

\caption{The items $[2n]$ are arranged in a table, where the $k\Th$ column is $B_k \defeq \{k, 2n+1-k\}$.
For $i \in N^1$, we have $v_i(B_k) = v_{i,k} + v_{i,2n+1-k} \le 1$ for all $k$.
However, $a \not\in N^1$. Hence, we look for the smallest \emph{shift} $\ell$ such that
$v_{a,k} + v_{a,2n+1-k+\ell} \le 1$ for all $k$.}
\label{fig:ell}
\end{figure*}

\begin{lemma}
\label{thm:S:tricky-bound}
$\sum_{k \in A^+} v_a(B'_k) < |A^+| + \min(\ell, |A^+|)/12$.
\end{lemma}
\begin{savedProof}{tricky-bound}
Let $S \in A^+$ be the set of $\min(\ell, |A^+|)$ smallest indices in $A^+$
and $L \in A^+$ be the set of $\min(\ell, |A^+|)$ largest indices in $A^+$.
Since $|A^+| \ge 1$ and $\ell \ge 1$, we get $|S| = |L| \ge 1$.
Note that
\begin{align*}
& \sum_{k \in A^+} v_a(B'_k)
= \left({\sum_{k \in S} v_{a,k} + \sum_{k \in L} v_{a,2\BFn-k+1}}\right)
    \TCNLA[\qquad]+ \left(\sum_{k \in A^+ \setminus S} v_{a,k} + \sum_{k \in A^+ \setminus L} v_{a,2\BFn-k+1}\right).
\end{align*}

By \cref{thm:pair-bound}, we get $v_{a,2\BFn-k+1} \le \frac{1}{3}$.
Since $v_{a,k} < 3/4\ifBetter{+\delta}{}$ and $|S| \ge 1$, we get
\begin{equation}
\ifBetter{\label{eqn:B:claim-1}}{\label{eqn:S:claim-1}}
\begin{aligned}
& \sum_{k \in S} v_{a,k} + \sum_{k \in L} v_{a,2\BFn-k+1}
\ifBetter{\TCNLA[\qquad]}{}< |S|\left(\frac{3}{4}\ifBetter{+\delta}{} + \frac{1}{3}\right)
= \addIfBetterP{\frac{13}{12}}{ + \delta}|S|.
\end{aligned}
\end{equation}

If $\ell \ge |A^+|$, then $|S| = |L| = |A^+|$, and we are done.
Now assume $\ell < |A^+|$. Then $|S| = |L| = \ell$.

Let $A^+ \defeq \{g_1, \ldots, g_{|A^+|}\}$ and $g_1 < \ldots < g_{|A^+|}$.
Then $A^+ \setminus S = \{g_{\ell+1}, \ldots, g_{|A^+|}\}$
and $A^+ \setminus L = \{g_1, \ldots, g_{|A^+|-\ell}\}$.
The idea is to pair the goods $g_{k+\ell}$ and $2\BFn-g_k+1$ (for $k \in [|A^+|-\ell]$)
and prove that their value is at most $1$ for agent $a$.

Since $g_{k+\ell} \ge g_k + \ell$, we get
$v_{a, g_{k+\ell}} + v_{a, 2\BFn-g_k+1} \le 1$ by definition of $\ell$. Hence,
\begin{equation}
\ifBetter{\label{eqn:B:claim-2}}{\label{eqn:S:claim-2}}
\begin{aligned}
& \sum_{k \in A^+ \setminus S} v_{a,k} + \sum_{k \in A^+ \setminus L} v_{a,2n-k+1}
\TCNLA= \sum_{k \in [|A^+|-\ell]} (v_{a, g_{k+\ell}} + v_{a, 2n-g_k+1}) \le |A^+| - \ell.
\end{aligned}
\end{equation}

\Cref{eqn:\secP:claim-1,eqn:\secP:claim-2} imply \cref{thm:\secP:tricky-bound}.
\end{savedProof}

\begin{lemma}
\label{thm:S:small-goods-bound}
$v_a([m] \setminus [2n]) > \ell/4$.
\end{lemma}
\begin{savedProof}{small-goods-bound}
By definition of $\ell$, there exists a good $k \in \{\ell, \ldots, \BFn\}$ such that
$v_{a,k} + v_{a,2\BFn-k+\ell} > 1$.
Hence, for all $j \in [k]$ and $t \le [2\BFn-k+\ell]$,
we have $v_{a,j} + v_{a,t} \ge v_{a,k} + v_{a,2\BFn-k+\ell} > 1$.

Let $P \defeq (P_1, \ldots, P_{\BFn})$ be an MMS partition of agent $a$.
Then, for $j \in [k]$ and $t \in [2\BFn-k+\ell]$, $j$ and $t$ cannot be in the same bundle in $P$,
since the instance is normalized.
In particular, no two goods from $[k]$ are in the same bundle in $P$.
Hence, assume \wLoG{} that $j \in P_j$ for all $j \in [k]$.
Thus, $[2\BFn-k+\ell] \setminus [k] \subseteq P_{k+1} \cup \ldots \cup P_{\BFn}$.

Bundles in $\{P_1, \ldots, P_k\}$ can only have goods from $[k]$,
$[2\BFn] \setminus [2\BFn-k+\ell]$, and $\ifBetter{D \cup }{}[m] \setminus [2\BFn]$.
There are $k-\ell$ goods in $[2\BFn] \setminus [2\BFn-k+\ell]$.
Hence, at least $\ell$ bundles in $\{P_1, \ldots, P_k\}$ have just 1 good from $[2\BFn]$.
Let $L$ be the indices of these bundles, i.e.,
$L \defeq \{t \in [k] \mid |P_t \cap [2\BFn]| = 1\}$. Then
\begin{align*}
v_a(\ifBetter{D \cup }{}[m] \setminus [2\BFn])
&\ge \sum_{j \in L} v_a(P_j \setminus \{j\})
\TCNLA= \sum_{j \in L} (v_a(P_j) - v_{a,j})
\ifBetter{\\ &}{\TCNLA}> \sum_{j \in L} \left(1 - \addIfBetterP{\frac{3}{4}}{ + \delta}\right)
    \reason{$v_{a,j} < 3/4\ifBetter{+\delta}{}$ by \cref{thm:vr-upper-bounds}}
\ifBetter{\TCNLA}{}= \ifBetter{|L|\left(\frac{1}{4}-\delta\right)}{\frac{|L|}{4}}
    \ge \ifBetter{\ell\left(\frac{1}{4}-\delta\right)}{\frac{\ell}{4}}.
\end{align*}
\end{savedProof}

\begin{lemma}
\label{thm:S:lower-1-2}
For all $i \in N^2$ and $k \in [n]$, $v_i(B_k) > 1/2$.
\end{lemma}
\begin{savedProof}{lower-1-2}
Fix an $i \in N^2$. Let $t$ be smallest such that $v_i(B_t) > 1$.
By \cref{thm:pair-bound}, $v_{i,t} > 2/3$.
Hence, for all $k \le t$,
\[ v_i(B_k) \ge v_{i,k} \ge v_{i,t} > \frac{2}{3} > \frac{1}{2}\ifBetter{-2\delta}{}. \]
Since $v_i(B_t) = v_{i, t} + v_{i, 2\BFn-t+1} > 1$ and
$v_{i,t} < 3/4\ifBetter{+\delta}{}$ (by \cref{thm:vr-upper-bounds}),
we get $v_{i, 2\BFn-t+1} > 1/4\ifBetter{-\delta}{}$.
For all $k > t$, we have $k < 2\BFn-k+1 < 2\BFn-t+1$. Hence,
\[ v_i(B_k) = v_{i,k} + v_{i, 2\BFn-k+1}
\ge 2 \cdot v_{i, 2\BFn-t+1} > \frac{1}{2}\ifBetter{-2\delta}{}.
\qedhere \]
\end{savedProof}

\begin{lemma}
\label{thm:S:n2-agents}
$U_A \cap N^2 = \emptyset$, i.e., every agent in $N^2$ gets a bag.
\end{lemma}
\begin{proof}
Assume for the sake of contradiction that $U_A \cap N^2 \neq \emptyset$.
Then, as discussed before, we fix an agent $a \in U_A \cap N^2$
and define $A^+$, $A^-$, $A^0$, and $\ell$.
\begin{align*}
n &= v_a([m]) = \sum_{k \in [n]} v_a(B'_k)
\TCNLA= \sum_{k \in A^-} v_a(B'_k) + \sum_{k \in A^+} v_a(B'_k) + \sum_{k \in A^0} v_a(B'_k)
\\ &< \frac{5}{6}|A^-| + \left(|A^+| + \frac{\ell}{12}\right) + |A^0|
    \reason{by \cref{thm:S:upper-5-6,thm:S:tricky-bound}}
\TCNLA= n + \frac{\ell}{12} - \frac{|A^-|}{6}
\end{align*}
Hence, $|A^-| < \ell/2$.

Now we show that there are enough goods in $[m] \setminus [2n]$ to fill the bags in $A^-$.
\begin{align*}
\frac{\ell}{4} &\le v_a([m] \setminus [2n])
    \reason{by \cref{thm:S:small-goods-bound}}
\\ &= \sum_{k \in A^-} (v_a(B'_k) - v_a(B_k))
    \reason{since $B_k' = B_k \subseteq [2n]$ for $k \in A^+ \cup A^0$}
\\ &< |A^-|\left(\frac{5}{6} - \frac{1}{2}\right)
    \reason{by \cref{thm:S:upper-5-6,thm:S:lower-1-2}}
\\ &= |A^-|\cdot\frac{1}{3} < \frac{\ell}{6},
    \reason{since $|A^-| < \ell/2$}
\end{align*}
which is a contradiction.
\end{proof}

By \cref{thm:S:n1-agents,thm:S:n2-agents}, we get that $U_A = \emptyset$, i.e.,
every agent gets a bag, and hence, $\bagFillHyp$'s output is $3/4$-MMS.

\section{\boldmath Better than \texorpdfstring{$3/4$}{3/4}-MMS}
\label{sec:betterSummary}

In this section, we give an overview of how to refine the techniques of \cref{sec:simple} to get
an algorithm that outputs a $(\frac{3}{4} + \imp)$-MMS allocation.
The details can be found in \coecref{sec:better}%
\ifNoAppendix{ of the full version of our paper \cite{akrami2023simplification}}.

\begin{theorem}
For any fair division instance with additive valuations,
a $(\frac{3}{4} + \imp)$-MMS allocation exists.
\end{theorem}

Algorithm $\approxMMS$ from \cref{sec:simple} does not work with $\alpha > 3/4$,
since $R_4(\alpha)$ may not be a valid reduction.
To fix this, we modify $R_4(\alpha)$ using the \emph{dummy goods} technique from \cite{garg2021improved}.

Consider the fair division instance $([n], [m], v)$.
When performing $R_4(\alpha)$, in addition to giving the goods $S_4 \defeq \{1, 2n+1\}$
to some agent $i$ for whom $v_i(S_4) \ge \alpha\MMS_i$,
we create a \emph{dummy good} $g$ where $v_j(g) \defeq \max(0, v_j(S_4) - \MMS_j)$
for each agent $j \neq i$.
With this change, $R_4(\alpha)$ becomes a valid reduction even for $\alpha > 3/4$.
See \coecref{sec:imp:fix-R4} for a proof.
Note that dummy goods are fictional, i.e.,
they exist solely to guide the valid reductions.
No agent is allocated a dummy good.

Formally, a fair division instance with dummy goods
is represented as a tuple $\Ical \defeq (N, M, v, D)$,
where $D$ is the set of dummy goods and $M$ is the set of non-dummy goods.
We can extend the concepts of \cref{sec:ordered} (ordered instance),
\cref{sec:valid-redn} (valid reductions), and \cref{sec:normalized} (normalized instance)
to instances with dummy goods. See \coecref{sec:imp:mods} for details.
In particular, instance $(N, M, v, D)$ is ordered iff $(N, M, v)$ is ordered,
and $(N, M, v, D)$ is normalized iff $(N, M \cup D, v)$ is normalized.

With these modifications, we can extend $\approxMMS$ to the case where $\alpha > 3/4$.
$\approxMMS$ first transforms the instance into an ordered, normalized, and
totally-$\alpha$-irreducible instance. Then it discards all the dummy goods
and allocates the remaining goods using the algorithm $\bagFillHyp$.
In \coecref{sec:imp:analysis}, we show that when $\alpha \le \frac{3}{4} + \imp$,
$\bagFill$ allocates a bag of value at least $\alpha$ to every agent.
Our proof is almost the same as that in \cref{sec:simple}.
The main difference is that the analogue of \cref{thm:S:n2-agents}
(\coecref{thm:B:n2-agents} in \coecref{sec:better})
involves more elaborate algebraic manipulations so that we can get tighter bounds.

\section{Tight Example}
\label{sec:tight-example}

We give an almost tight example for our algorithm
and Garg and Taki's \shortcite{garg2021improved} algorithm.
We show that these algorithms' output on this example is not better than
$(\frac{3}{4} + \frac{3}{8n-4})$-MMS.

\begin{example}
\label{ex:1}
Consider a fair division instance with $n$ agents and $m = 3n-1$ goods.
All agents have the same valuation function $u$, where
\[ u(j) \defeq \begin{cases}
\displaystyle \frac{2n-1-\floor{(j-1)/2}}{4n-2} & \textrm{ if } j \le 2n
\\ \displaystyle \frac{n}{4n-2} & \textrm{ if } j > 2n
\end{cases}. \]
\end{example}

\begin{lemma}
\label{thm:ex1-mms}
\Cref{ex:1} is normalized.
\end{lemma}
\begin{proof}
Let $M_1 \defeq \{1, 2\}$ and for $i \in [n-1]$, let $M_{i+1} \defeq \{i+2, 2n+1-i, 2n+i\}$.
Then for any $i \neq j$, $M_i \cap M_j = \emptyset$.
Also, $u(M_1) = u(1) + u(2) = 1$ and for each $i \in [n-1]$,
\begin{align*}
& (4n-2)u(M_{i+1})
\\ &= (4n-2)(u(i+2) + u(2n+1-i) + u(2n+i))
\\ &= \Big(2n-1-\floor[\Big]{\frac{i+1}{2}}\Big) + \Big(2n-1-\floor[\Big]{\frac{2n-i}{2}}\Big) + n
\\ &= (2n-1-\ceil{i/2}) + (n-1+\ceil{i/2}) + n
\\ &= 4n-2.
\qedhere \end{align*}
\end{proof}

Define the $\mmsScore$ of an allocation as the maximum $\alpha$ such that
it is an $\alpha$-MMS allocation.
Formally, for an allocation $A \defeq (A_1, \ldots, A_n)$,
\[ \mmsScore(A) \defeq \min_{i=1}^n \frac{v_i(A_i)}{\MMS_i}. \]

\begin{theorem}
\label{thm:ex1-gt}
Let $\Ical$ be the fair division instance of \cref{ex:1}.
Let $S_1 \defeq \{1\}$, $S_2 \defeq \{n, n+1\}$, $S_3 \defeq \{2n-1, 2n, 2n+1\}$,
$S_4 \defeq \{1, 2n+1\}$.
Consider a fair division algorithm that either outputs $\bagFillHyp(\Ical, \alpha)$ for some $\alpha$,
or allocates the set $S_k$, for some $k \in [4]$, to an agent $i$,
and allocates the remaining goods to the remaining agents in an unspecified way.
Let $A$ be the allocation output by this algorithm. Then
\[ \mmsScore(A) \le \frac{3n}{4n-2} = \frac{3}{4} + \frac{3}{8n-4}. \]
\end{theorem}
\begin{proof}
$u(S_1) = 1/2$, $u(S_2) = u(S_4) = (3n-1)/(4n-2)$, and $u(S_3) = 3n/(4n-2)$.
Hence, if the algorithm allocates $S_k$ to an agent $i$, for some $k \in [4]$,
then that agent will get a bundle of value at most $3n/(4n-2)$.

Now suppose that the algorithm outputs $\bagFill(\Ical, \alpha)$.
Every bag initially has value $\tau \defeq (3n-1)/(4n-2)$.
If $\alpha \le \tau$, then no bag receives any more items,
and each agent gets a bag of value $\tau$.
If $\alpha > \tau$, then we run out of goods and $\bagFill$ fails (i.e., returns \texttt{null}),
since there are $n$ bags but only $n-1$ goods in $[m] \setminus [2n]$.
\end{proof}

\section{Conclusion}
\label{sec:conclusion}

In fair division of indivisible goods, MMS is one of the most popular notions of fairness,
and determining (tight lower and upper bounds on) the maximum $\alpha$ for which
$\alpha$-MMS allocations are guaranteed to exist is an important open problem.

To gain a better understanding of this problem, we thoroughly studied
Garg and Taki's \shortcite{garg2021improved} algorithm for obtaining $3/4$-MMS allocations.
We considerably simplified its analysis and our techniques helped
improve the best-known MMS approximation factor to $\frac{3}{4} + \imp$.
Furthermore, we presented a tight example that reveals a fundamental barrier
towards improving the MMS approximation guarantee using techniques in \cite{garg2021improved}.

\appendix

\section{\boldmath Existence of \texorpdfstring{$(\frac{3}{4} + \imp)$-\mms}{(3/4+\impText)-MMS}~Allocations}
\label{sec:better}

\renewcommand{\ifBetter}[2]{#1}

In this section, we show how to refine the techniques of \cref{sec:simple} to get
an algorithm that outputs a $(\frac{3}{4} + \imp)$-MMS allocation.
We would like to use the algorithm $\approxMMS$ from \cref{sec:simple}
with $\alpha$ slightly more than $3/4$. However,
$R_4(\alpha)$ may not be a valid reduction when $\alpha > 3/4$.
Hence, we will slightly modify the $\approxMMS$ algorithm so that it works even for $\alpha > 3/4$.
This modification is based on the \emph{dummy goods} technique from \cite{garg2021improved}.

\subsection{\boldmath Fixing \texorpdfstring{$R_4$}{R\_4}}
\label{sec:imp:fix-R4}

Let us recall the proof of $R_4(\alpha)$ being a valid reduction when $\alpha \le 3/4$.

\begin{lemma}
Consider the fair division instance $([n], [m], v)$,
where $v_{i,1} \ge \ldots \ge v_{i,m}$ for each agent $i$
and the instance is $R_1(\alpha)$ and $R_3(\alpha)$ irreducible.
Let $S_4 \defeq \{1, 2n+1\}$ and $v_i(S_4) \ge \alpha\MMS_{v_i}^n([m])$.
Then giving $S_4$ to agent $i$ is a valid reduction if $\alpha \le 3/4$.
\end{lemma}
\begin{proof}
Fix an agent $j \neq i$ and let $P$ be $j$'s MMS partition.
Let $\beta \defeq \MMS_{v_j}^n([m])$. We get two cases.

\textbf{Case 1}: goods $1$ and $2n+1$ belong to the same bundle $P_k$.
\\ Then redistribute the goods $P_k \setminus S_4$ among the remaining bundles in $P$,
and give $S_4$ to agent $i$.
The bundles $\{P_t \mid t \neq k\}$ give us a partition of $[m] \setminus S_4$
among agents $[n] \setminus \{i\}$, and each bundle has value at least $\beta$.
Hence, $\MMS_{v_j}^{n-1}([m] \setminus S_4) \ge \beta$, so $R_4(\alpha)$ is a valid reduction.
Note that this case works for all $\alpha \le 1$.

\textbf{Case 2}: goods $1$ and $2n+1$ belong to different bundles of $P$.
Let these bundles be $P_{n-1}$ and $P_n$ (\wLoG{}).
Define $Q \defeq (Q_1, \ldots, Q_{n-1})$ as
\[ Q_k \defeq \begin{cases}P_k & \textrm{ if } k \in [n-2]
\\ P_{n-1} \cup P_n \setminus S_4 & \textrm{ if } k = n-1
\end{cases}. \]
To show that $R_4(\alpha)$ is a valid reduction, we just need to show that $v_j(Q_{n-1}) \ge \beta$.
$R_1(\alpha)$-irreducibility implies $v_{j,1} < \alpha\beta$.
$R_3(\alpha)$-irreducibility implies $v_{j,2n+1} < \alpha\beta/3$ by \cref{thm:vr-upper-bounds}.
Hence, $v_j(S_4) < 4\alpha\beta/3$. Hence,
\begin{align*}
v_j(Q_{n-1}) &= v_j(P_{n-1}) + v_j(P_n) - v_j(S_4)
\TCNLA> \beta + \beta - 4\alpha\beta/3 = \beta(2 - 4\alpha/3).
\end{align*}
When $\alpha \le 3/4$, we get $v_j(Q_{n-1}) \ge \beta$.
Hence, $R_4(\alpha)$ is a valid reduction.
\end{proof}

Note that case 2 of the proof fails when $\alpha > 3/4$ because $Q_{n-1}$ doesn't have enough value.
To remedy this, we add a \emph{dummy good} of value $\max(0, v_j(S_4) - \beta)$ to $Q_{n-1}$.

Note that these dummy goods are fictional. They exist solely to guide the valid reductions.
Once we obtain an ordered, normalized, totally-$\alpha$-irreducible instance,
we throw away all the dummy goods in the beginning of $\bagFill$.
Hence, no agent receives a dummy good.

\subsection{Fair Division with Dummy Goods}
\label{sec:imp:mods}

We now formally define dummy goods, and what it means for an instance containing dummy goods
to be ordered, normalized, and irreducible.

A fair division instance is represented as a tuple $\Ical \defeq (N, M, v, D)$,
where $N$ is the set of agents, $M$ is the set of (non-dummy) goods,
$D$ is the set of dummy goods, and $v_{i,g}$ is agent $i$'s valuation for good $g \in M \cup D$.
An allocation $A \defeq (A_1, \ldots, A_{|N|})$ in $\Ical$ is a partition of $M$, i.e.,
dummy goods are not allocated but all other goods are allocated.

An allocation $A$ is $\alpha$-MMS for $\Ical$ if $v_i(A_i) \ge \MMS_{v_i}^{|N|}(M \cup D)$,
i.e., agents take dummy goods into consideration when computing their MMS values.
When the fair division instance $(N, M, v, D)$ is clear from context,
we write $\MMS_i$ instead of $\MMS_{v_i}^{|N|}(M \cup D)$.

An instance $(N, M, v, D)$ is ordered if $(N, M, v)$ is ordered.
We extend $\toOrdered$ to this setting by simply ignoring the dummy goods. Formally,
$\toOrdered((N, M, v, D))$ is defined as $(N, [|M|], \vhat, D)$,
where $(N, [|M|], \vhat) \defeq \toOrdered((N, M, v))$ and
$\vhat_{i,g} \defeq v_{i,g}$ for all $g \in D$.

$(N, M, v, D)$ is said to be normalized if $(N, M \cup D, v)$ is normalized.
$\normalize$ works analogously: it computes each agent's MMS partition of $M \cup D$,
and scales all goods (both dummy and non-dummy) so that each partition has value 1.

\begin{definition}[Valid reduction with dummy goods]
In a fair division instance $(N, M, v, D)$, suppose we give the goods $S \subseteq M$
to agent $i$ and create dummy goods $T$.
Then we are left with a new instance $(N \setminus \{i\}, M \setminus S, v, D \cup T)$.
Such a transformation is called a \emph{valid} $\alpha$-\emph{reduction} if
both of these conditions hold:
\begin{tightenum}
\item $v_i(S) \ge \alpha\MMS_{v_i}^{|N|}(M \cup D)$.
\item $\MMS_{v_j}^{|N|-1}((M \setminus S) \cup D \cup T) \ge \MMS_{v_j}^{|N|}(M \cup D)$
    for all $j \in N \setminus \{i\}$.
\end{tightenum}
\end{definition}

The reduction rules remain the same (as in \cref{defn:redn-rules}),
except that we modify $R_4(\alpha)$ when $\alpha > 3/4$:
if we give goods $S_4$ to agent $i$, then we create a dummy good $g$,
where $v_{j,g} \defeq \max(0, v_j(S_4) - \MMS_{v_j}^{|N|}(M \cup D))$ for each agent $j$.
We also modify the algorithm $\reduceHyp_{\alpha}$ to use the new $R_4(\alpha)$.

\subsection{Algorithm}
\label{sec:imp:analysis}

The high-level algorithm $\approxMMS$ (\cref{algo:gt}) is the same as that in \cref{sec:simple}.
The only changes are that we modify $\toOrderedHyp$, $\normalizeHyp$, and $\reduceHyp_{\alpha}$
as described in \cref{sec:imp:mods}.
$\bagFillHyp$ discards all dummy goods and then proceeds as before.

Let there be $n$ agents in the original fair division instance
and let $n'$ be the number of remaining agents after all reduction rules are applied
(i.e., we performed $n-n'$ valid reductions).
Let the input to $\bagFill$ be the instance $\Ical \defeq ([n'], [m], v, D)$,
which is ordered, normalized, and totally-$\alpha$-irreducible.
Then $n \ge n' + |D|$. Let $\delta \defeq \alpha - 3/4$.

\begin{lemma}
\label{thm:dummy-bound}
For every agent $i$ and every dummy good $g$, we have $v_{i,g} < 4\delta/3$.
\end{lemma}
\begin{proof}
Assume $\delta > 0$, since otherwise there are no dummy goods.
Suppose $g$ was created when $R_4(\alpha)$ was applied to the instance
$\Icaltild \defeq (\Ntild, \Mtild, \vtild, \Dtild)$
to get $(\Ntild \setminus \{j\}, \Mtild \setminus S_4, \vtild, \Dtild \cup \{g\})$.
Let $\betatild \defeq \MMS_{\vtild_i}^{|\Ntild|}(\Mtild \cup \Dtild)$.
Since $\Icaltild$ is $R_1(\alpha)$ and $R_3(\alpha)$ irreducible,
by \cref{thm:vr-upper-bounds}, we get that $\vtild_i(S_4) < (4/3)\alpha\betatild$.
Hence, $\vtild_{i,g} \defeq \max(0, \vtild_i(S_4) - \betatild) < (4/3)\betatild\delta$.

Just before we $\normalize$ the instance, let $\beta$ be agent $i$'s MMS value.
Then $\beta \ge \betatild$, since the reduction rules are valid reductions.
After we $\normalize$ the instance, we scale down every good's value by $\beta$.
Hence, after $\normalize$, $g$'s value to agent $i$ is
$v_{i,g} = \vtild_{i,g} / \beta \le \vtild_{i,g} / \betatild < (4/3)\delta$.
\end{proof}

For $k \in [n']$, let $B_k \defeq \{k, 2n'+1-k\}$ be the initial contents of the $k\Th$ bag
and $B'_k$ be the $k\Th$ bag's contents after $\bagFill$ terminates.
Just like \cref{sec:simple}, we consider two groups of agents:
$N^1 \defeq \{i \in [n'] \mid \forall k \in [n'], v_i(B_k) \le 1\}$
and $N^2 \defeq [n'] \setminus N^1$.
Let $U_A$ be the set of agents that didn't receive a bag when $\bagFill$ terminated.
We first show that all agents in $N^1$ receive a bag, i.e., $U_A \cap N^1 = \emptyset$.
Then we show that $U_A \cap N^2 = \emptyset$.
Together, these facts establish that $\bagFill$ terminates successfully,
and hence its output is $\alpha$-MMS.

In this section, many lemmas and their proofs are very similar to those in \cref{sec:simple}.
Hence, their proofs have been moved to \cref{sec:missing-proofs}.

\begin{restatable}{lemma}{rthm:B:upper-1}
\label{thm:B:upper-1}
Let $i$ be any agent.
For all $k \in [n']$, if $v_i(B_k) \le 1$, then $v_i(B'_k) \le 1 + 4\delta/3$.
\end{restatable}
\movedProofNote

\begin{lemma}
\label{thm:B:n1-agents}
If $\delta \le 3/(16n-4)$, then $U_A \cap N^1 = \emptyset$, i.e., every agent in $N^1$ gets a bag.
\end{lemma}
\begin{proof}
For the sake of contradiction, assume $U_A \cap N^1 \neq \emptyset$.
Hence, $\exists i \in U_A \cap N^1$. Also, for some $t \in [n']$, the $t\Th$ bag is unallocated.
We have $v_i(B_t') < 3/4 + \delta$. Therefore,
\begin{align*}
n' &= v_i([m] \cup D)
\\ &= v_i(D) + v_i(B_t') + \sum_{k \in [n'] \setminus \{t\}} v_i(B'_k)
    \reason{since $[m] = \bigcup_{k \in [n']} B'_k$}
\\ &< |D|\left(\frac{4\delta}{3}\right) + \left(\frac{3}{4} + \delta\right)
    + (n'-1)\left(1 + \frac{4\delta}{3}\right)
    \reason{by \cref{thm:dummy-bound,thm:B:upper-1}}
\\ &= n' - \frac{1}{4} + \delta \left(\frac{4(|D| + n') - 1}{3}\right)
\\ &\le n' - \frac{1}{4} + \delta \frac{4n-1}{3}
    \reason{since $|D| + n' \le n$}
\\ &\le n',
    \reason{since $\delta \le 3/(16n-4)$}
\end{align*}
which is a contradiction. Hence, $U_A \cap N^1 = \emptyset$.
\end{proof}

Now we prove that $\approxMMS$ allocates a bag to all agents $i \in N^2$.
Assume for the sake of contradiction that $U_A \neq \emptyset$.
Let $a$ be a fixed agent in $U_A$. By \cref{thm:S:n1-agents}, $a \in N^2$.

Let $A^+ \defeq \{k \in [n'] \mid v_a(B_k) > 1\}$,
$A^- \defeq \{k \in [n'] \mid v_a(B_k) < 3/4\}$,
and $A^0 \defeq \{k \in [n'] \mid 3/4 \le v_a(B_k) \le 1\}$.
We will try to get upper bounds on $v_a(B_k')$ for each of the cases
$k \in A^+$, $k \in A^-$, and $k \in A^0$.
Note that $n' = |A^+| + |A^-| + |A^0|$.
Also, $n' \in A^-$ since the instance is $R_2(3/4+\delta)$-irreducible,
and $|A^+| \ge 1$ since $a \in N^2$.

Let $\ell$ be the smallest such that for all $k \in [\ell+1, n']$,
$v_{a,k} + v_{a, 2n'-k+1 + \ell} \le 1$.
Note that $\ell \ge 1$, since $a \in N^2$.
The following five lemmas have their proofs moved to \cref{sec:missing-proofs}.

\begin{restatable}{lemma}{rthm:B:upper-1-12}
\label{thm:B:upper-1-12}
$i \in N^2 \implies v_{i, 2n'+1} < 1/12 + \delta$.
\end{restatable}

\begin{restatable}{lemma}{rthm:B:upper-5-6}
\label{thm:B:upper-5-6}
$\forall k \in A^-$, $v_a(B'_k) < 5/6 + 2\delta$.
\end{restatable}

\begin{restatable}{lemma}{rthm:B:tricky-bound}
\label{thm:B:tricky-bound}
$\displaystyle \sum_{k \in A^+} v_i(B'_k)
    < |A^+| + \min(|A^+|, \ell)\Big(\frac{1}{12}+\delta\Big)$.
\end{restatable}

\begin{restatable}{lemma}{rthm:B:small-goods-bound}
\label{thm:B:small-goods-bound}
$v_i(D \cup [m] \setminus [2n']) > \ell(1/4 - \delta)$.
\end{restatable}

\begin{restatable}{lemma}{rthm:B:lower-1-2}
\label{thm:B:lower-1-2}
For all $i \in N^2$ and $k \in [n']$, $v_i(B_k) > 1/2-2\delta$.
\end{restatable}

\begin{restatable}{lemma}{rthm:B:n2-agents}
\label{thm:B:n2-agents}
If either $n = 2$ and $\delta \le 1/12$,
or if $\delta \le \min(\frac{1}{36}, \frac{3}{16n-4})$,
then $U_A \cap N^2 = \emptyset$, i.e., every agent in $N^2$ gets a bag.
\end{restatable}
\begin{proof}
Assume for the sake of contradiction that $U_A \cap N^2 \neq \emptyset$.
Then, as discussed before, we fix an agent $a \in U_A \cap N^2$
and define $A^+$, $A^-$, $A^0$, and $\ell$.

Let $S \defeq \{k \in A^- \mid v_a(B_k') < 3/4 + \delta\}$.
Then $|S| \ge 1$, since $i$ didn't get a bag.
\begin{align*}
n' &= v_a([m] \cup D) = \sum_{k=1}^{n'} v_a(B_k')
\\ &= \sum_{k \in A^+} v_a(B_k') + \sum_{k \in A^0} v_a(B_k')
    \TCNLA[\qquad]+ \sum_{k \in A^- \setminus S} v_a(B_k') + \sum_{k \in S} v_a(B_k')
\\ &< \left(|A^+| + \min(|A^+|, \ell)\left(\frac{1}{12}+\delta\right)\right) + |A^0|
    \TCNLA[\qquad]+ (|A^-|-|S|)\left(\frac{5}{6} + 2\delta\right) + |S|\left(\frac{3}{4} + \delta\right)
    \reason{by \cref{thm:B:tricky-bound,thm:B:upper-5-6}}
\\ &= n' + \min(|A^+|, \ell)\left(\frac{1}{12}+\delta\right)
    \TCNLA[\qquad]+ (|A^-|-|S|)\left(-\frac{1}{6} + 2\delta\right) + |S|\left(-\frac{1}{4} + \delta\right)
    \reason{since $n' = |A^+| + |A^0| + |A^-|$}
\end{align*}
Then $\forall \ellhat \in \{\ell, |A^+|\}$, we get
\begin{align}
& \nonumber 0 < \ellhat\left(\frac{1}{12}+\delta\right) + (|A^-|-|S|)\left(-\frac{1}{6} + 2\delta\right)
    \TCNLA[\qquad]\nonumber+ |S|\left(-\frac{1}{4} + \delta\right)
\\ &\label{eq:delta-lb1}\implies \delta > \frac{1}{12}\;
    \frac{2(|A^-|-|S|) + 3|S| - \ellhat}{\ellhat + |S| + 2(|A^-|-|S|)}.
\end{align}
Using \cref{thm:B:small-goods-bound,thm:dummy-bound}, we get
\[ v_a([m] \setminus [2n']) > \ell\left(\frac{1}{4} - \delta\right) - \frac{4|D|}{3}\delta. \]
Using \cref{thm:B:upper-5-6,thm:B:lower-1-2}, we get
\begin{align*}
& v_a([m] \setminus [2n'])
\TCNLA= \sum_{k \in A^- \setminus S} v_a(B_k') + \sum_{k \in S} v_a(B_k') - \sum_{k \in A^-} v_a(B_k)
\\ &< (|A^-|-|S|)\left(\frac{5}{6} + 2\delta\right) + |S|\left(\frac{3}{4} + \delta\right)
    \TCNLA[\qquad]- |A^-|\left(\frac{1}{2} - 2\delta\right)
\\ &= (|A^-|-|S|)\left(\frac{1}{3} + 4\delta\right) + |S|\left(\frac{1}{4} + 3\delta\right).
\end{align*}
Hence,
\begin{align}
& \nonumber\ell\left(\frac{1}{4} - \delta\right) - \frac{4|D|}{3}\delta < v_a([m] \setminus [2n'])
    \TCNLA[\qquad]\nonumber< (|A^-|-|S|)\left(\frac{1}{3} + 4\delta\right) + |S|\left(\frac{1}{4} + 3\delta\right)
\\ &\label{eq:delta-lb2}\implies \delta > \frac{1}{12}\;
    \frac{3\ell - 3|S| - 4(|A^-|-|S|)}{\ell + 4|D|/3 + 3|S| + 4(|A^-|-|S|)}.
\end{align}
Let $x \defeq |S|-1$ and $y \defeq |A^-|-|S|$.
Then $x \ge 0$ and $y \ge 0$, since $|S| \ge 1$ and $S \subseteq A^-$.
By \eqref{eq:delta-lb1} and \eqref{eq:delta-lb2}, $\forall \ellhat \in \{\ell, |A^+|\}$, we get
\begin{equation}
\label{eq:delta-lb}
\begin{aligned}
\delta > \frac{1}{12}\max\bigg(&
    \frac{3x + 2y + (3 - \ellhat)}{x + 2y + 1 + \ellhat},
    \TCNLA[\quad]\frac{3\ell - 3 - 3x - 4y}{\ell + 4|D|/3 + 3 + 3x + 4y}
\bigg).
\end{aligned}
\end{equation}

The following useful inequality is proved in \cref{sec:missing-proofs}.
\begin{restatable}{claim}{rthmTrinomialRatio}
\label{thm:trinomial-ratio}
If $a_2$, $b_2$, and $c_2$ are positive, then $\forall x \ge 0$, $\forall y \ge 0$,
\[ \frac{a_1x + b_1y + c_1}{a_2x + b_2y + c_2}
\ge \min\left(\frac{a_1}{a_2}, \frac{b_1}{b_2}, \frac{c_1}{c_2}\right). \]
\end{restatable}

Let $p$, $q$, and $r$ be arbitrary constants whose values we will decide later.
If $\ell \ge px + qy + r$, then
\begin{align*}
12\delta &> \frac{3\ell - 3 - 3x - 4y}{\ell + 4|D|/3 + 3 + 3x + 4y}  \reason{by \cref{eq:delta-lb}}
\\ &\ge \frac{3(px + qy + r) - 3 - 3x - 4y}{(px + qy + r) + 4|D|/3 + 3 + 3x + 4y}
\\ &= \frac{(3p-3)x + (3q-4)y + (3r-3)}{(p+3)x + (q+4)y + (r+3+4|D|/3)}
\\ &\ge \min\left(\frac{3p-3}{p+3}, \frac{3q-4}{q+4}, \frac{3r-3}{r+3+4|D|/3}\right).
    \reason{by \cref{thm:trinomial-ratio}}
\end{align*}
If $\ell \le px + qy + r$, then
\begin{align*}
12\delta &> \frac{3x + 2y + 3 - \ell}{x + 2y + 1 + \ell}  \reason{by \cref{eq:delta-lb}}
\\ &\ge \frac{3x + 2y + 3 - (px + qy + r)}{x + 2y + 1 + (px + qy + r)}
\\ &= \frac{(3-p)x + (2-q)y + (3-r)}{(p+1)x + (q+2)y + (r+1)}
\\ &\ge \min\left(\frac{3-p}{p+1}, \frac{2-q}{q+2}, \frac{3-r}{r+1}\right).
    \reason{by \cref{thm:trinomial-ratio}}
\end{align*}
Hence, $\forall \ell$,
\begin{align*}
12\delta \ge \min\bigg(&\frac{3p-3}{p+3}, \frac{3-p}{p+1}, \frac{3q-4}{q+4}, \frac{2-q}{q+2},
    \TCNLA[\quad]\frac{3r-3}{r+3+4|D|/3}, \frac{3-r}{r+1}\bigg).
\end{align*}
On setting $p = 2$, $q = (\sqrt{17}-1)/2$, and $r = 2$, we get
\[ 12\delta \ge \min\left(\frac{1}{3}, \frac{9}{4|D| + 15}\right). \]

If $|A^+| \le 2$, then
\begin{align*}
& 12\delta > \frac{3x + 2y + (3-|A^+|)}{x + 2y + 1 + |A^+|}
    \TCNLA[\quad]\ge \min\left(3, 1, \frac{3-|A^+|}{1+|A^+|}\right)
    = \begin{cases}1 & \textrm{ if } |A^+| = 1 \\ 1/3 & \textrm{ if } |A^+| = 2\end{cases}.
\end{align*}
Note that if $n = 2$, then $n' = 2$ and $|A^+| = 1$.

If $|A^+| \ge 3$, then $n' \ge 4$, so $|D| \le n - n' \le n - 4$.
Hence,
\[ 12\delta > \min\left(\frac{1}{3}, \frac{9}{4(n-4) + 15}\right)
= \min\left(\frac{1}{3}, \frac{9}{4n-1}\right). \]
Hence, if someone in $N^2$ doesn't get a bag, then $\delta$ is sufficiently large.
So, if we pick $\delta$ to be sufficiently small, everyone in $N^2$ will get a bag.
\end{proof}

By \cref{thm:B:n1-agents,thm:B:n2-agents}, we get that $U_A = \emptyset$, i.e.,
every agent gets a bag, and hence, $\bagFill$'s output is $(3/4 + \delta)$-MMS
when $\delta \le \imp$.

\section{Missing Proofs}
\label{sec:missing-proofs}

\callMacro{rthm:B:upper-1}*
\recallProof{upper-1}

\callMacro{rthm:B:upper-1-12}*
\recallProof{upper-1-12}

\callMacro{rthm:B:upper-5-6}*
\recallProof{upper-5-6}

\callMacro{rthm:B:tricky-bound}*
\recallProof{tricky-bound}

\callMacro{rthm:B:small-goods-bound}*
\recallProof{small-goods-bound}

\callMacro{rthm:B:lower-1-2}*
\recallProof{lower-1-2}

\rthmTrinomialRatio*
\begin{proof}
\begin{align*}
& \forall x \ge 0, \forall y \ge 0, \frac{a_1x + b_1y + c_1}{a_2x + b_2y + c_2} \ge \frac{a_1}{a_2}
\\ &\iff \forall x \ge 0, \forall y \ge 0,
    \TCNLA[\qquad]\frac{(a_2b_1 - a_1b_2)y + (a_2c_1 - a_1c_2)}{a_2(a_2x + b_2y + c_2)} \ge 0
\\ &\iff a_2b_1 - a_1b_2 \ge 0 \textrm{ and } a_2c_1 - a_1c_2 \ge 0
\\ &\iff \frac{b_1}{b_2} \ge \frac{a_1}{a_2} \textrm{ and } \frac{c_1}{c_2} \ge \frac{a_1}{a_2}
\\ &\iff \frac{a_1}{a_2} = \min\left(\frac{a_1}{a_2}, \frac{b_1}{b_2}, \frac{c_1}{c_2}\right).
\end{align*}
\begin{align*}
\\ &\forall x \ge 0, \forall y \ge 0, \frac{a_1x + b_1y + c_1}{a_2x + b_2y + c_2} \ge \frac{b_1}{b_2}
\\ &\iff \forall x \ge 0, \forall y \ge 0,
    \TCNLA[\qquad]\frac{(a_1b_2 - a_2b_1)x + (b_2c_1 - b_1c_2)}{b_2(a_2x + b_2y + c_2)} \ge 0
\\ &\iff a_1b_2 - a_2b_1 \ge 0 \textrm{ and } b_2c_1 - b_1c_2 \ge 0
\\ &\iff \frac{a_1}{a_2} \ge \frac{b_1}{b_2} \textrm{ and } \frac{c_1}{c_2} \ge \frac{b_1}{b_2}
\\ &\iff \frac{b_1}{b_2} = \min\left(\frac{a_1}{a_2}, \frac{b_1}{b_2}, \frac{c_1}{c_2}\right).
\end{align*}
\begin{align*}
\\ &\forall x \ge 0, \forall y \ge 0, \frac{a_1x + b_1y + c_1}{a_2x + b_2y + c_2} \ge \frac{c_1}{c_2}
\\ &\iff \forall x \ge 0, \forall y \ge 0,
    \TCNLA[\qquad]\frac{(a_1c_2 - a_2c_1)x + (b_1c_2 - b_2c_1)y)}{c_2(a_2x + b_2y + c_2)} \ge 0
\\ &\iff a_1c_2 - a_2c_1 \ge 0 \textrm{ and } b_1c_2 - b_2c_1 \ge 0
\\ &\iff \frac{a_1}{a_2} \ge \frac{c_1}{c_2} \textrm{ and } \frac{b_1}{b_2} \ge \frac{c_1}{c_2}
\\ &\iff \frac{c_1}{c_2} = \min\left(\frac{a_1}{a_2}, \frac{b_1}{b_2}, \frac{c_1}{c_2}\right).
\qedhere \end{align*}
\end{proof}


\begin{thebibliography}{10}

\bibitem{amanatidis2017truthful}
Georgios Amanatidis, Georgios Birmpas, George Christodoulou, and Evangelos
  Markakis.
\newblock Truthful allocation mechanisms without payments: Characterization and
  implications on fairness.
\newblock In {\em ACM Conference on Economics and Computation}, pages 545--562,
  2017.
\newblock \href {https://doi.org/10.1145/3033274.3085147}
  {\path{doi:10.1145/3033274.3085147}}.

\bibitem{amanatidis2016truthful}
Georgios Amanatidis, Georgios Birmpas, and Evangelos Markakis.
\newblock On truthful mechanisms for maximin share allocations.
\newblock In {\em International Joint Conference on Artificial Intelligence},
  pages 31--37, 2016.
\newblock URL: \url{https://www.ijcai.org/Abstract/16/012}.

\bibitem{amanatidis2017approximation}
Georgios Amanatidis, Evangelos Markakis, Afshin Nikzad, and Amin Saberi.
\newblock Approximation algorithms for computing maximin share allocations.
\newblock {\em ACM Transactions on Algorithms (TALG)}, 13(4):1--28, 2017.
\newblock \href {https://doi.org/10.1145/3147173} {\path{doi:10.1145/3147173}}.

\bibitem{aziz2019strategyproof}
Haris Aziz, Bo~Li, and Xiaowei Wu.
\newblock Strategyproof and approximately maxmin fair share allocation of
  chores.
\newblock In {\em International Joint Conference on Artificial Intelligence},
  pages 60--66, 2019.
\newblock \href {https://doi.org/10.24963/ijcai.2019/9}
  {\path{doi:10.24963/ijcai.2019/9}}.

\bibitem{aziz2017algorithms}
Haris Aziz, Gerhard Rauchecker, Guido Schryen, and Toby Walsh.
\newblock Algorithms for max-min share fair allocation of indivisible chores.
\newblock In {\em AAAI Conference on Artificial Intelligence}, 2017.
\newblock \href {https://doi.org/10.1609/aaai.v31i1.10582}
  {\path{doi:10.1609/aaai.v31i1.10582}}.

\bibitem{babaioff2021fair}
Moshe Babaioff, Tomer Ezra, and Uriel Feige.
\newblock Fair-share allocations for agents with arbitrary entitlements.
\newblock In {\em ACM Conference on Economics and Computation}, pages 127--127,
  2021.
\newblock \href {https://doi.org/10.1145/3465456.3467559}
  {\path{doi:10.1145/3465456.3467559}}.

\bibitem{babaioff2022fair}
Moshe Babaioff and Uriel Feige.
\newblock Fair shares: Feasibility, domination and incentives.
\newblock In {\em ACM Conference on Economics and Computation}, page 435.
  Association for Computing Machinery, 2022.
\newblock \href {https://doi.org/10.1145/3490486.3538286}
  {\path{doi:10.1145/3490486.3538286}}.

\bibitem{barman2018groupwise}
Siddharth Barman, Arpita Biswas, Sanath Krishnamurthy, and Yadati Narahari.
\newblock Groupwise maximin fair allocation of indivisible goods.
\newblock In {\em AAAI Conference on Artificial Intelligence}, 2018.
\newblock \href {https://doi.org/10.1609/aaai.v32i1.11463}
  {\path{doi:10.1609/aaai.v32i1.11463}}.

\bibitem{barman2019fair}
Siddharth Barman, Ganesh Ghalme, Shweta Jain, Pooja Kulkarni, and Shivika
  Narang.
\newblock Fair division of indivisible goods among strategic agents.
\newblock In {\em International Conference on Autonomous Agents and Multi-Agent
  Systems}, page 1811–1813, 2019.
\newblock URL:
  \url{https://www.ifaamas.org/Proceedings/aamas2019/pdfs/p1811.pdf}.

\bibitem{barman2020approximation}
Siddharth Barman and Sanath~Kumar Krishnamurthy.
\newblock Approximation algorithms for maximin fair division.
\newblock {\em ACM Transactions on Economics and Computation (TEAC)},
  8(1):1--28, 2020.
\newblock \href {https://doi.org/10.1145/3381525} {\path{doi:10.1145/3381525}}.

\bibitem{bei2022price}
Xiaohui Bei, Ayumi Igarashi, Xinhang Lu, and Warut Suksompong.
\newblock The price of connectivity in fair division.
\newblock {\em SIAM Journal on Discrete Mathematics}, 36(2):1156--1186, 2022.
\newblock \href {https://doi.org/10.1137/20M1388310}
  {\path{doi:10.1137/20M1388310}}.

\bibitem{biswas2018fair}
Arpita Biswas and Siddharth Barman.
\newblock Fair division under cardinality constraints.
\newblock In {\em IJCAI}, pages 91--97, 2018.
\newblock \href {https://doi.org/10.24963/ijcai.2018/13}
  {\path{doi:10.24963/ijcai.2018/13}}.

\bibitem{bouveret2016characterizing}
Sylvain Bouveret and Michel Lema{\^\i}tre.
\newblock Characterizing conflicts in fair division of indivisible goods using
  a scale of criteria.
\newblock {\em Autonomous Agents and Multi-Agent Systems}, 30(2):259--290,
  2016.
\newblock \href {https://doi.org/10.1007/s10458-015-9287-3}
  {\path{doi:10.1007/s10458-015-9287-3}}.

\bibitem{budish2011combinatorial}
Eric Budish.
\newblock The combinatorial assignment problem: Approximate competitive
  equilibrium from equal incomes.
\newblock {\em Journal of Political Economy}, 119(6):1061--1103, 2011.
\newblock \href {https://doi.org/10.1086/664613} {\path{doi:10.1086/664613}}.

\bibitem{chaudhury2021little}
Bhaskar~Ray Chaudhury, Telikepalli Kavitha, Kurt Mehlhorn, and Alkmini
  Sgouritsa.
\newblock A little charity guarantees almost envy-freeness.
\newblock {\em SIAM Journal on Computing}, 50(4):1336--1358, 2021.
\newblock \href {https://doi.org/10.1137/20M1359134}
  {\path{doi:10.1137/20M1359134}}.

\bibitem{deuermeyer1982scheduling}
Bryan~L. Deuermeyer, Donald~K. Friesen, and Michael~A. Langston.
\newblock Scheduling to maximize the minimum processor finish time in a
  multiprocessor system.
\newblock {\em SIAM Journal on Algebraic Discrete Methods}, 3(2):190--196,
  1982.
\newblock \href {https://doi.org/10.1137/0603019} {\path{doi:10.1137/0603019}}.

\bibitem{farhadi2019fair}
Alireza Farhadi, Mohammad Ghodsi, Mohammad~Taghi Hajiaghayi, Sebastien Lahaie,
  David Pennock, Masoud Seddighin, Saeed Seddighin, and Hadi Yami.
\newblock Fair allocation of indivisible goods to asymmetric agents.
\newblock {\em Journal of Artificial Intelligence Research}, 64:1--20, 2019.
\newblock \href {https://doi.org/10.1613/jair.1.11291}
  {\path{doi:10.1613/jair.1.11291}}.

\bibitem{feige2022improved}
Uriel Feige and Alexey Norkin.
\newblock Improved maximin fair allocation of indivisible items to three
  agents.
\newblock {\em arXiv:2205.05363}, 2022.
\newblock \href {https://doi.org/10.48550/ARXIV.2205.05363}
  {\path{doi:10.48550/ARXIV.2205.05363}}.

\bibitem{feige2021tight}
Uriel Feige, Ariel Sapir, and Laliv Tauber.
\newblock A tight negative example for {MMS} fair allocations.
\newblock In {\em International Conference on Web and Internet Economics},
  pages 355--372. Springer, 2021.
\newblock \href {https://doi.org/10.1007/978-3-030-94676-0_20}
  {\path{doi:10.1007/978-3-030-94676-0_20}}.

\bibitem{garg2019approximating}
Jugal Garg, Peter McGlaughlin, and Setareh Taki.
\newblock Approximating maximin share allocations.
\newblock In {\em Symposium on Simplicity in Algorithms (SOSA 2019)},
  volume~69, pages 20:1--20:11, 2018.
\newblock \href {https://doi.org/10.4230/OASIcs.SOSA.2019.20}
  {\path{doi:10.4230/OASIcs.SOSA.2019.20}}.

\bibitem{garg2021improved}
Jugal Garg and Setareh Taki.
\newblock An improved approximation algorithm for maximin shares.
\newblock {\em Artificial Intelligence}, page 103547, 2021.
\newblock \href {https://doi.org/10.1016/j.artint.2021.103547}
  {\path{doi:10.1016/j.artint.2021.103547}}.

\bibitem{GatesGD20}
Vael Gates, Thomas~L. Griffiths, and Anca~D. Dragan.
\newblock How to be helpful to multiple people at once.
\newblock {\em Cognitive Science}, 44, 2020.

\bibitem{ghodsi2018fair}
Mohammad Ghodsi, MohammadTaghi HajiAghayi, Masoud Seddighin, Saeed Seddighin,
  and Hadi Yami.
\newblock Fair allocation of indivisible goods: Improvements and
  generalizations.
\newblock In {\em ACM Conference on Economics and Computation}, pages 539--556,
  2018.
\newblock \href {https://doi.org/10.1145/3219166.3219238}
  {\path{doi:10.1145/3219166.3219238}}.

\bibitem{gourves2019maximin}
Laurent Gourv{\`e}s and J{\'e}r{\^o}me Monnot.
\newblock On maximin share allocations in matroids.
\newblock {\em Theoretical Computer Science}, 754:50--64, 2019.
\newblock \href {https://doi.org/10.1016/j.tcs.2018.05.018}
  {\path{doi:10.1016/j.tcs.2018.05.018}}.

\bibitem{hosseini2021guaranteeing}
Hadi Hosseini and Andrew Searns.
\newblock Guaranteeing maximin shares: Some agents left behind.
\newblock In {\em International Joint Conference on Artificial Intelligence},
  pages 238--244, 2021.
\newblock \href {https://doi.org/10.24963/ijcai.2021/34}
  {\path{doi:10.24963/ijcai.2021/34}}.

\bibitem{huang2021algorithmic}
Xin Huang and Pinyan Lu.
\newblock An algorithmic framework for approximating maximin share allocation
  of chores.
\newblock In {\em ACM Conference on Economics and Computation}, pages 630--631,
  2021.
\newblock \href {https://doi.org/10.1145/3465456.3467555}
  {\path{doi:10.1145/3465456.3467555}}.

\bibitem{kurokawa2016can}
David Kurokawa, Ariel~D Procaccia, and Junxing Wang.
\newblock When can the maximin share guarantee be guaranteed?
\newblock In {\em AAAI Conference on Artificial Intelligence}, 2016.
\newblock URL:
  \url{https://www.aaai.org/ocs/index.php/AAAI/AAAI16/paper/viewPaper/12282}.

\bibitem{kurokawa2018fair}
David Kurokawa, Ariel~D Procaccia, and Junxing Wang.
\newblock Fair enough: Guaranteeing approximate maximin shares.
\newblock {\em Journal of the ACM (JACM)}, 65(2):1--27, 2018.
\newblock \href {https://doi.org/10.1145/3140756} {\path{doi:10.1145/3140756}}.

\bibitem{li2021fair}
Zhentao Li and Adrian Vetta.
\newblock The fair division of hereditary set systems.
\newblock {\em ACM Transactions on Economics and Computation (TEAC)},
  9(2):1--19, 2021.
\newblock \href {https://doi.org/10.1145/3434410} {\path{doi:10.1145/3434410}}.

\bibitem{procaccia2014fair}
Ariel~D Procaccia and Junxing Wang.
\newblock Fair enough: Guaranteeing approximate maximin shares.
\newblock In {\em ACM Conference on Economics and Computation}, pages 675--692,
  2014.
\newblock \href {https://doi.org/10.1145/2600057.2602835}
  {\path{doi:10.1145/2600057.2602835}}.

\bibitem{truszczynski2020maximin}
Miroslaw Truszczynski and Zbigniew Lonc.
\newblock Maximin share allocations on cycles.
\newblock {\em Journal of Artificial Intelligence Research}, 69:613--655, 2020.
\newblock \href {https://doi.org/10.1613/jair.1.11702}
  {\path{doi:10.1613/jair.1.11702}}.

\bibitem{woeginger1997polynomial}
Gerhard~J Woeginger.
\newblock A polynomial-time approximation scheme for maximizing the minimum
  machine completion time.
\newblock {\em Operations Research Letters}, 20(4):149--154, 1997.
\newblock \href {https://doi.org/10.1016/S0167-6377(96)00055-7}
  {\path{doi:10.1016/S0167-6377(96)00055-7}}.

\end{thebibliography}
\end{document}